\documentclass[journal]{IEEEtran}
\usepackage{mathrsfs}
\usepackage{mathptmx}
\usepackage{amsmath}
\usepackage{amsfonts}
\usepackage{amssymb}
\usepackage{amsthm}
\usepackage{array}
\usepackage{subfigure}
\usepackage{textcomp}
\usepackage{float}
\usepackage{graphicx}
\usepackage{cite}
\usepackage{bm}
\usepackage{arydshln}
\usepackage{xcolor}
\usepackage{psfrag}
\usepackage{multirow}
\usepackage{booktabs}
\usepackage{chngcntr}
\usepackage[backref]{hyperref}
\usepackage[OT2,T1]{fontenc}

\newtheorem*{problem}{Problem Statement}

\newtheorem{defi}{Definition}
\newtheorem{thm}{Theorem}

\newtheorem{rem}{Remark}

\newtheorem{lem}{Lemma}

\begin{document}

\title{Learning and Control from Similarity Between Heterogeneous Systems: A Behavioral Approach}
	
\author{Chenchao Wang and Deyuan Meng, \emph{Senior Member, IEEE}
\thanks{Chenchao Wang is with the School of Automation Science and Electrical Engineering, Beihang University (BUAA), Beijing 100191, P. R. China, and also with the Seventh Research Division, Beihang University (BUAA), Beijing 100191, P. R. China (e-mail: chenchaow1999@163.com).}
\thanks{Deyuan Meng is with the School of Automation Science and Electrical Engineering, Beihang University (BUAA), Beijing 100191, P. R. China, with the State Key Laboratory of CNS/ATM, Beijing 100191, P. R. China, and also with the Seventh Research Division, Beihang University (BUAA), Beijing 100191, P. R. China (e-mail: dymeng@buaa.edu.cn).}
	}

\maketitle
\begin{abstract}
This paper proposes basic definitions of similarity and similarity indexes between heterogeneous linear systems and presents a similarity-based learning control strategy. By exploring geometric properties of admissible behaviors of linear systems, the similarity indexes between two admissible behaviors of heterogeneous systems are defined as the principal angles between their subspace components, and an efficient strategy for calculating the similarity indexes is developed. By leveraging the similarity indexes, a similarity-based learning control strategy is proposed via projection techniques. With the application of the similarity-based learning control strategy, host system can efficiently accomplish the same tasks by leveraging the successful experience of guest system, without the necessity to repeat the trial-and-error process experienced by the guest system.
\end{abstract}
	
\begin{IEEEkeywords}
Similarity index; similarity-based learning; admissible behavior; experience; projection technique.
\end{IEEEkeywords}
	
\section{Introduction}\label{sec:introduction}
	
Learning-based control has received considerable attention and emerged as one of the most promising methodologies in the field of intelligent control in the past decades. Inspired by the human learning process, learning-based control values the experience generated during the control process and employs it to improve future control performance. Learning-based control possesses a rich history, resulting in the developments of various influential intelligent control frameworks \cite{ZJiang2020FTSC}. For instance, iterative learning control (ILC) that learns from past iterations was first proposed in 1980s to better manipulate mechanical robots \cite{SArimoto1984JRS}; Reinforcement learning (RL) attempts to obtain the optimal control policy by leveraging the experience from trial-and-error interactions in a dynamic environment \cite{KArulkumaran2017IEEESPM}; Neural-network-based control (NNBC) repetitively selects a mini-batch from training data and minimizes the corresponding cost function to drive the estimated weighting parameters of approximated models approaching optimal values \cite{SAlbawi2017ICET}. All of these control strategies learn from the past experience and modify the controllers/policies/parameters, and further achieve desired control performances and possess lower dependency on the accurate model knowledge. Thanks to its advantages of low model-dependency and high control accuracy, learning-based control has been widely applied in numerous industrial fields such as robots \cite{SSchaal2010IEEERAM}, autonomous vehicles \cite{SKuutti2020TITS}, intelligent transportation systems \cite{HDong2010IEEECSM} and so on.
	
	Nevertheless, those aforementioned learning-based control strategies only focus on how individual systems leverage their own past experience. If we scrutinize human learning behavior, it is evident that besides personal past experience, external successful experience, especially that from individuals similar to oneself, also plays an essential role. An illustrative example is that learning skills from a highly skilled teacher can often be more effective than repetitive trial-and-error. A straightforward inspiration is whether a host system can also benefit from the successful experience from other guest systems to enhance its own control performances. Unfortunately, to our knowledge, few control strategies concentrate on leveraging the successful experience of other systems \cite{AJVan2004TAC}. Furthermore, due to the model differences among heterogeneous systems, it is necessary to determine which experience is more valuable when receiving external successful experience from multiple guest systems.
	
	Building upon the aforementioned discussions, this paper proposes the definitions of similarity and similarity indexes, and presents an innovative similarity-based learning control strategy for heterogeneous linear time-varying (LTV) systems. This strategy leverages the successful experience of the guest systems to enhance  control performance of the host system. We begin from exploring helpful geometric properties of the admissible behaviors of LTV systems. Subsequently, to qualitatively and quantitatively clarify the benefits of the experience of the guest system for the control of host system, we introduce the concepts of similarity and similarity indexes between two admissible behaviors of heterogeneous systems, respectively. Moreover, the criterion for verifying similarity and strategy for calculating the similarity indexes are presented. By leveraging the similarity indexes and the geometric characteristics of the admissible behaviors, we employ the projection techniques to formulate the similarity-based learning control strategy. As a consequence, the host system can benefit from the successful control experience of the guest system, and the effectiveness of similarity-based learning control depends on the similarity indexes between the host and guest systems. The mechanism of the proposed similarity-based learning control strategy is depicted in Fig. \ref{fig-mechanism}.
	
	\begin{figure}[!ht]
		\centering
		{\includegraphics[width=0.9\columnwidth]{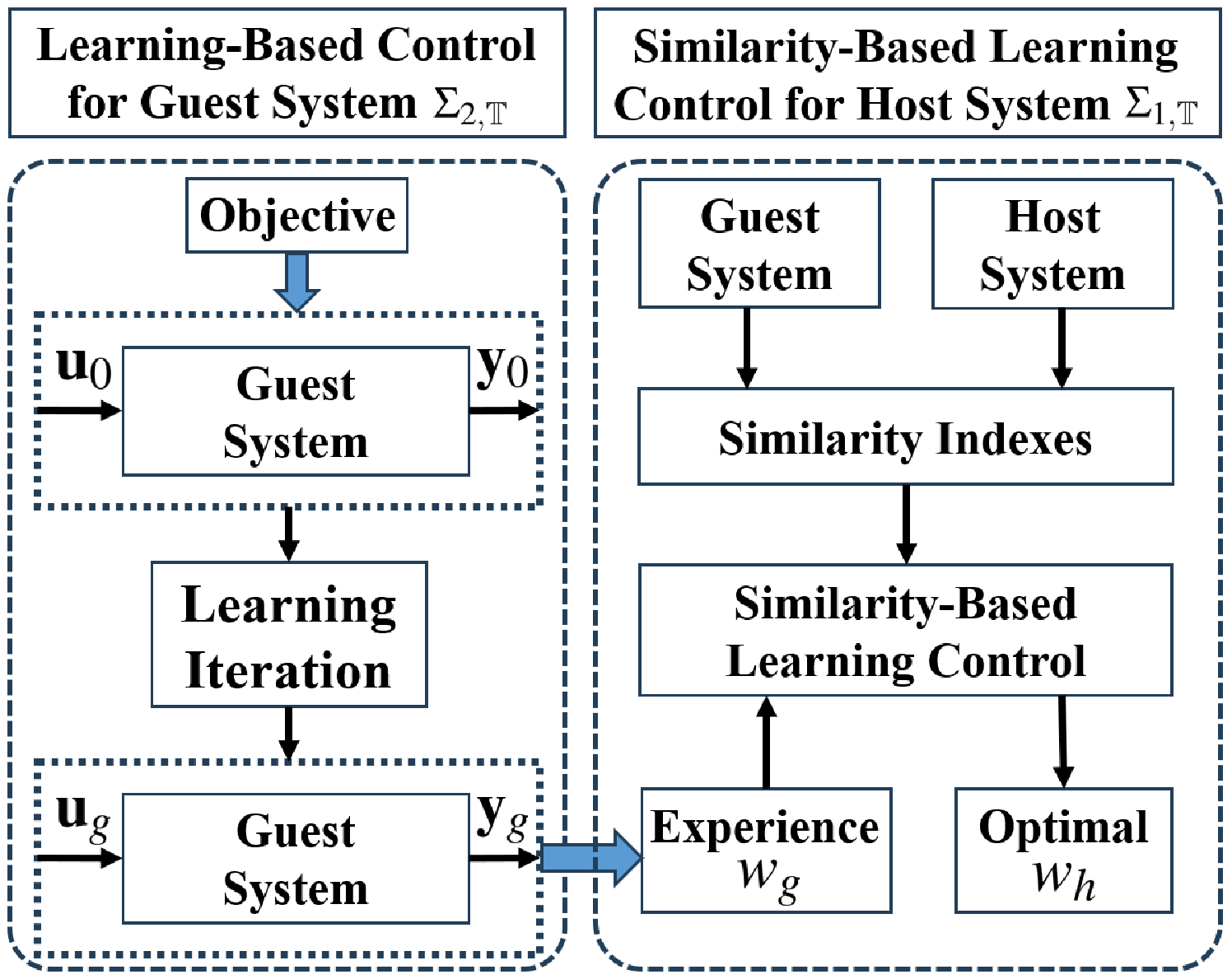}}
		\caption{The mechanism of the similarity-based learning control.}
		\label{fig-mechanism}
	\end{figure}
	
	The reminder of this paper is organized as follows. We present some preliminaries of the admissible behaviors of LTV systems and formulate the similarity-based learning control problems in Section \ref{sec:Preliminaries and Problem Statement}. We introduce the fundamental definitions regarding similarity and similarity indexes, and provide a criterion for the similarity and a strategy for calculating the similarity indexes in Section \ref{sec:similarity and similarity indexes}. By integrating the geometric properties of the admissible behaviors and similarity indexes, we develop a similarity-based learning control strategy via projection techniques in Section \ref{sec:similarity-based learning via projection}. In Section \ref{sec:Simulation Examples}, we provide illustrative examples to demonstrate the effectiveness of the proposed similarity-based learning control strategy. Finally, we summarize our contributions in Section \ref{sec:Conclusions and future works}.

	\emph{Notations:} Let $\mathbb{Z}_N=\{0,1,\cdots,N\}$ and $\mathbb{Z}_+=\{0,1,2,\cdots\}$. Let $\mathbb{R}$ be the set of all real numbers, and $\mathbb{R}^n$ involves all $n$-dimensional vectors whose entries locate in $\mathbb{R}$. For any matrix $A$, its transpose and kernel space are denoted as $A^{\rm T}$ and $\text{ker}\left(A\right)$, respectively. For arbitrary vectors $a,b\in\mathbb{R}^n$, the inner product $\langle a,b\rangle$ refers to $a^{\rm T}b$, and the induced norm is defined as $\left\| a\right\|=\sqrt{\langle a,a\rangle}$. The identity and null matrices with appropriate dimensions are denoted as $I$ and $0$, respectively. Given $s_1,s_2,\cdots,s_n\in\mathbb{R}$, the symbol $\text{diag}(s_1,s_2,\cdots,s_n)$ represents the diagonal matrix whose diagonal entries are $s_1,s_2,\cdots,s_n$.
	\section{Preliminaries and Problem Statement} \label{sec:Preliminaries and Problem Statement}
	Consider two heterogeneous discrete LTV systems with identical time duration $\mathbb{T}$ in the form of
	\begin{equation} \label{eq-model}
		\Sigma_{i,\mathbb{T}}: \left\{
			\begin{aligned}
				x_i(t+1)&=A_i(t)x_i(t)+B_i(t)u_i(t)\\
				y_i(t)&=C_i(t)x_i(t)+D_i(t)u_i(t)
			\end{aligned},\ t\in\mathbb{T},\ i\in\{1,2\}.
			\right.
	\end{equation}
	It is worth mentioning that the results presented in this paper can be equally implemented to the cases where $D_i(t)\equiv0$ for all $t\in\mathbb{T}$, and we introduce the input-output pass-through matrix $D_i(t)$ purely for the sake of generality.
	Without loss of generality, we assume that $\mathbb{T}:=\mathbb{Z}_{T-1}$. Indeed, when the initial time is not zero, we can always shift the time duration to $\mathbb{T}$ by adjusting the model parameters $\left(A_i(t),B_i(t),C_i(t),D_i(t)\right)$ without affecting theoretical results developed in this paper. Throughout the paper, $i=1$ and $i=2$ refer to the host and guest systems, respectively. The input and output are denoted as $u_i(t)\in\mathbb{R}^{n_u}$ and $y_i(t)\in\mathbb{R}^{n_y}$, respectively. To establish a connection between the input and output, the internal state, denoted as $x_i(t)\in\mathbb{R}^{n_x}$, is introduced, and we introduce the following lifted supervectors
	\begin{equation}
		\begin{aligned}
			\nonumber
			\mathbf{u}_i=
			\begin{bmatrix}
				u_i^{\rm T}(0),\ u_i^{\rm T}(1),\ \cdots,\ u_i^{\rm T}(T-1)
			\end{bmatrix}^{\rm T},\\
			\mathbf{y}_i=
			\begin{bmatrix}
				y_i^{\rm T}(0),\ y_i^{\rm T}(1),\ \cdots,\ y_i^{\rm T}(T-1)
			\end{bmatrix}^{\rm T}, \\
			\mathbf{x}_i=
			\begin{bmatrix}
				x_i^{\rm T}(0),\ x_i^{\rm T}(1),\ \cdots,\ x_i^{\rm T}(T-1)
			\end{bmatrix}^{\rm T}.
		\end{aligned}
	\end{equation}
	If there exists a $\mathbf{x}_i$ such that $\left(\mathbf{u}_i,\mathbf{y}_i,\mathbf{x}_i\right)$ satisfies (\ref{eq-model}), then $w_i=\text{col}(\mathbf{u}_i,\mathbf{y}_i)\in\mathbb{R}^{n_wT}$ where $n_w=n_u+n_y$ is called as a \emph{T-length trajectory} of $\Sigma_{i,\mathbb{T}}$. The input-output transfer characteristic is captured by the \emph{behavior} of $\Sigma_{i,\mathbb{T}}$, which is defined by
	\begin{equation}
		\mathcal{B}_i=\{w_i\in\mathbb{R}^{n_wT}|\exists\mathbf{x}_i\text{ such that }(\mathbf{u}_i,\mathbf{y}_i,\mathbf{x}_i)\text{ satisfies } (\ref{eq-model})\}.\nonumber
	\end{equation}
	Nevertheless, in addition to the input-output transfer characteristic, the initial state $x_i(0)$, which characterizes the initially stored energy in the dynamical system $\Sigma_{i,\mathbb{T}}$, also influences the output response \cite{RLozano2013book}. Without loss of generality, we assume that throughout this paper, the initial state of system $\Sigma_{i,\mathbb{T}}$ is $x_i(0)=x_i$. Once we impose restrictions on the initial state, the trajectory and behavior of $\Sigma_{i,\mathbb{T}}$ will be subject to additional constraints. The \emph{$T$-length admissible trajectories}, denoted by $w_{i,x_i}$, refer to those $T$-length trajectories who start from the initial state $x_i$. Correspondingly, the \emph{admissible behavior} can be represented as
	\begin{equation} \nonumber
		\mathcal{B}_{i,x_{i}}=\{w_{i,x_i}\in\mathbb{R}^{n_wT}\left|w_{i,x_i}\in\mathcal{B}_i\text{ and } x_i(0)=x_i\right.\}.
	\end{equation}
	The following lemma is employed to illustrate the geometric property of the admissible behavior.
	\begin{lem} \rm \label{lem-DTaffineset}
		For the discrete LTV system $\Sigma_{i,\mathbb{T}}$ with the initial state $x_i(0)=x_i$, the admissible behavior $\mathcal{B}_{i,x_i}$ is an affine set.
	\end{lem}
	\begin{proof}
		By introducing the input-output transfer matrix $G_i\in\mathbb{R}^{n_yT\times n_uT}$ and the initial state-output transfer matrix $L_i\in\mathbb{R}^{n_yT\times n_x}$, both of which can be readily constructed by leveraging the model knowledge $\left(A_i(t),B_i(t),C_i(t),D_i(t)\right)$, we obtain the relationship between the lifted input and output as \cite{DABristow2006IEEECSM}
		\begin{equation} \label{eq-iorelationship}
			\mathbf{y}_i=G_i\mathbf{u}_i+L_ix_i(0).
		\end{equation}
		Let $w_{i,x_i}$, $v_{i,x_i}\in\mathcal{B}_{i,x_i}$ be arbitrary admissible trajectories of $\Sigma_{i,\mathbb{T}}$, then $w_{i,x_i}$ and $v_{i,x_i}$ must satisfy the following non-homogeneous linear algebraic equations
		\begin{equation} \label{eq-affinecombination}
			\begin{aligned}
				\begin{bmatrix}
					-G_i,\  I
				\end{bmatrix}w_{i,x_i}&=L_ix_i,\\
				\begin{bmatrix}
					-G_i,\  I
				\end{bmatrix}v_{i,x_i}&=L_ix_i
			\end{aligned}
		\end{equation}
		when the initial state is chosen as $x_i(0)=x_i$.
		Given any integer $\alpha\in\mathbb{R}$, the affine combination of $w_{i,x_i}$ and $v_{i,x_i}$ is defined as $\alpha w_{i,x_i}+(1-\alpha)v_{i,x_i}$. From (\ref{eq-affinecombination}), it is directly concluded that
		\begin{equation} \nonumber
			\begin{bmatrix}
				-G_i,\ I
			\end{bmatrix}
			\left(\alpha w_{i,x_i}+(1-\alpha)v_{i,x_i}\right)=L_ix_i,\ \forall \alpha\in\mathbb{R}.
		\end{equation}
		Or equivalently, any affine combination of $w_{i,x_i}$ and $v_{i,x_i}$ is still an admissible trajectory of $\Sigma_{i,\mathbb{T}}$. Therefore, the admissible behavior $\mathcal{B}_{i,x_i}$ constitutes an affine set.
	\end{proof}
	\begin{rem} \rm \label{rem-decomposition}
		By noticing the fact that an affine set can always be decomposed into the sum of a linear subspace and an offset, the admissible behavior $\mathcal{B}_{i,x_i}$ can be further represented as
		\begin{equation} \label{eq-decomposition}
			\begin{aligned}
				\mathcal{B}_{i,x_{i}}&=\mathcal{W}_i+w_{i,\text{off}},\\
				\mathcal{W}_i&=\text{span}\left(H_i\right),\\
				H_i&=\begin{bmatrix}
					\alpha_{i,1},\ \alpha_{i,2},\ \cdots,\ \alpha_{i,n_uT}
				\end{bmatrix}
			\end{aligned}
		\end{equation}
		where $\mathcal{W}_i=\text{ker}\left(\left[-G_i,\ I\right]\right)$ is a subspace in $\mathbb{R}^{n_wT}$ and is of dimension $n_uT$, which is exactly the number of the free input channels \cite{SBoyd2004book}. The vectors $\{\alpha_{i,1},\alpha_{i,2},\cdots,\alpha_{i,n_uT}\}$ are a set of unit orthogonal bases of the subspace $\mathcal{W}_{i}$, that is,
		\begin{equation} \nonumber
			\begin{aligned}
				\langle \alpha_{i,j},\alpha_{i,k}\rangle&=0,\ \forall j\neq k, \ \forall j,k\in\mathbb{Z}_{n_uT}\backslash\{0\}, \ i\in\{1,2\};\\
				\langle \alpha_{i,j},\alpha_{i,j}\rangle&=1,\ \forall j\in\mathbb{Z}_{n_uT}\backslash\{0\},\ i\in\{1,2\}.
			\end{aligned}
		\end{equation}
		Moreover, the offset component $w_{i,\text{off}}$ is essentially a special solution of (\ref{eq-iorelationship}) under $x_i(0)=x_i$, and it can be directly chosen as $w_{i,\text{off}}=\text{col}\left(0_{n_uT},L_ix_i\right)$ for simplicity.
	\end{rem}
	\begin{rem} \rm \label{rem-datadriven}
		Even when the model knowledge of systems $\Sigma_{i,\mathbb{T}}$ is inaccessible, the decomposition mentioned in Remark \ref{rem-decomposition} can be alternatively accomplished via data-driven methods. By employing the persisting excitation condition or its variants, one can reformulate the subspace and offset components in (\ref{eq-decomposition}) via the Hankel matrices constructed from offline sampled data \cite{JCWillems2005SCL}. This also ensures the proposed similarity-based learning control strategy is less model-dependent.
	\end{rem}
	After introducing the preliminaries regarding the admissible behaviors, the similarity-based learning control problems are formulated as follows.
	\begin{problem} \rm
		For the host system $\Sigma_{1,\mathbb{T}}$ with initial state $x_1$ and the guest system $\Sigma_{2,\mathbb{T}}$ with initial state $x_2$, let their admissible behaviors be given as $\mathcal{B}_{1,x_1}$ and $\mathcal{B}_{2,x_2}$, respectively. This paper focuses on solving the following two problems:
		\begin{enumerate}
			\item \label{pro-1}The definitions of the similarity and similarity indexes between two admissible behaviors of heterogeneous systems need to be presented. Moreover, a criterion for clarifying similarity and a strategy for efficiently calculating similarity indexes need to be developed;
			\item \label{pro-2}Suppose that the guest system $\Sigma_{2,\mathbb{T}}$ has accomplished the control tasks (such as tracking, stabilization, regulation, etc.) and obtained the desired admissible trajectory $w_g\in\mathcal{B}_{2,x_2}$ through some learning-based control strategy (such as ILC, RL, NNBC, etc.). When the host system $\Sigma_{1,\mathbb{T}}$ is confronted with the same tasks, we need to present a similarity-based learning control strategy by employing the successful experience of $\Sigma_{2,\mathbb{T}}$. As a result, we will find a solution $w_h\in\mathcal{B}_{1,x_1}$ such that the difference $\left\| w_g-w_h\right\|$ is minimized.
		\end{enumerate}
	\end{problem}
	\section{Similarity and Similarity Indexes} \label{sec:similarity and similarity indexes}
	Regarding the host system $\Sigma_{1,\mathbb{T}}$ with initial state $x_1(0)=x_1$ and the guest system $\Sigma_{2,\mathbb{T}}$ with initial state $x_2(0)=x_2$, the objective of this section is to present the definitions of similarity and similarity indexes between two admissible behaviors $\mathcal{B}_{1,x_1}$ and $\mathcal{B}_{2,x_2}$, which may qualitatively and quantitatively assess the benifits of the guest system's experience to the host system, respectively. Moreover, the similarity criterion and the strategy for efficiently calculating the similarity indexes are developed in this section. The concept of similarity between $\mathcal{B}_{1,x_1}$ and $\mathcal{B}_{2,x_2}$ and its criterion are first introduced as follows.
	\begin{defi} \rm
		The admissible behaviors $\mathcal{B}_{1,x_1}$ and $\mathcal{B}_{2,x_2}$ are said to be similar if $\mathcal{B}_{1,x_1}\cap\mathcal{B}_{2,x_2}\neq\emptyset$.
	\end{defi}
	\begin{lem}\rm \label{lem-criterion}
		The admissible behaviors $\mathcal{B}_{1,x_1}$ and $\mathcal{B}_{2,x_2}$ are similar if and only if there exists a vector $w_{com}$ such that
		\begin{equation} \label{eq-criterion}
			\begin{bmatrix}
				-G_1,\  I\\
				-G_2,\  I
			\end{bmatrix}w_{com}=
			\begin{bmatrix}
				L_1x_1\\
				L_2x_2
			\end{bmatrix}.
		\end{equation}
	\end{lem}
	\begin{rem}\rm
		Although the similarity is a rather loose concept, it can provide assistance in some special cases for solving Problem \ref{pro-2}). Specifically, when desired admissible trajectory $w_g$ satisfies (\ref{eq-criterion}), the successful experience of the guest system can be directly adopted to guarantee the control performance of the host system without any adjustments. Moreover, similarity serves as the basis for proposing the definition of similarity indexes, since there is no necessity to further investigate the similarity indexes for dissimilar admissible behaviors.
	\end{rem}
	To further quantitatively assess the benefits of the successful experience of the guest system for the host system, we need to introduce a more refined concept, namely similarity indexes. Before that, we first introduce the concepts of principal angles and principal vectors between two subspaces.
	\begin{defi} \rm  \cite{PAAbsil2006lLAA} \label{defi-principalangle}
		For two subspaces $\mathcal{W}_1\subset \mathbb{R}^{n_wT}$ and $\mathcal{W}_2\subset\mathbb{R}^{n_wT}$ with $\text{dim}\left(\mathcal{W}_1\right)=\text{dim}\left(\mathcal{W}_2\right)=n_uT$, the principal angles
		\begin{equation}\nonumber
			\Theta\left(\mathcal{W}_1,\mathcal{W}_2\right)=
			\begin{bmatrix}
				\theta_1,\theta_2,\cdots,\theta_{n_uT}
			\end{bmatrix},\ \theta_k\in\left[0,\frac{\pi}{2}\right],\ k\in\mathbb{Z}_{n_uT}\backslash\{0\}
		\end{equation}
		between $\mathcal{W}_1$ and $\mathcal{W}_2$ are recursively defined by
		\begin{equation}\nonumber
			s_k=\cos\left(\theta_k\right)=\max_{x\in\mathcal{W}_1}\max_{y\in\mathcal{W}_2}\langle x,y\rangle=\langle x_k,y_k\rangle
		\end{equation}
		subject to
		\begin{equation}\nonumber
			\left\| x\right\|=\left\| y\right\|=1,\ \langle x,x_i\rangle=0,\ \langle y,y_i\rangle=0,\ i\in\mathbb{Z}_{k-1}\backslash\{0\}.
		\end{equation}
		Moreover, the vectors $\{x_1,x_2,\cdots,x_{n_uT}\}$ and $\{y_1,y_2,\cdots,y_{n_uT}\}$ are called the principal vectors associated with $\mathcal{W}_1$ and $\mathcal{W}_2$.
	\end{defi}
	\begin{rem} \rm \label{rem-explanation2principalangles}
		The cosine values of the principal angles obtained according to Definition \ref{defi-principalangle} is monotonically non-increasing, that is, $\cos(\theta_k)\geq\cos(\theta_{k+1}), \forall k\in\mathbb{Z}_{n_uT-1}\backslash\{0\}$. This seemingly sophisticated definition may be explained as following steps:
		\begin{enumerate}
			\item When calculating $\cos(\theta_1)$, the unit vectors $x_1\in\mathcal{W}_1$ and $y_1\in\mathcal{W}_2$ can be freely chosen to ensure that the inner product $\langle x_1,y_1\rangle$ is maximized;
			\item When calculating $\cos(\theta_2)$, the orthogonal complements of $x_1$ and $y_1$, denoted as $\mathcal{W}_1-\text{span}(x_1)$ and $\mathcal{W}_2-\text{span}(y_1)$, respectively, need to be first calculated. Afterward, the unit vectors $x_2\in\mathcal{W}_1-\text{span}(x_1)$ and $y_2\in\mathcal{W}_2-\text{span}(y_1)$ can be freely chosen to ensure that $\langle x_2,y_2\rangle$ is maximized;
			\item When calculating $\cos(\theta_k)$, we need to calculate the orthogonal complements $\mathcal{W}_1-\text{span}(x_1,x_2,\cdots,x_{k-1})$ and $\mathcal{W}_2-\text{span}(y_1,y_2,\cdots,y_{k-1})$. Afterward, the unit vectors $x_k\in\mathcal{W}_1-\text{span}(x_1,x_2,\cdots,x_{k-1})$ and $y_k\in\mathcal{W}_2-\text{span}(y_1,y_2,\cdots,y_{k-1})$ can be freely selected to ensure that $\langle x_k,y_k\rangle$ is maximized.
		\end{enumerate}
	\end{rem}
	\begin{rem} \rm\label{rem-principalangle}
		Obviously, $\{x_1,x_2,\cdots,x_{n_uT}\}$ (or $\{y_1,y_2,\cdots,y_{n_uT}\}$) is a special set of the unit orthogonal bases for subspace $\mathcal{W}_1$ (or $\mathcal{W}_2$). The definition of principal angles depends on the specific definition of inner product. Additionally, the principal angles are independent on the offset components $w_{1,\text{off}}$ and $w_{2,\text{off}}$.  We may understand this claim from a geometric perspective. For two similar admissible behaviors $\mathcal{B}_{1,x_1}$ and $\mathcal{B}_{2,x_2}$, since the offset components only cause corresponding translations of the hyperplanes determined by $\mathcal{W}_i$, they do not violate the principal angles between two hyperplanes.
	\end{rem}
	As described in Remark \ref{rem-decomposition}, the admissible behaviors $\mathcal{B}_{i,x_i}$ can be decomposed into the sum of subspace and offset components. Therefore, the principal angles between two subspace components can serve as a powerful tool for quantitatively characterizing the similarity between two admissible behaviors. The definition of similarity indexes is given as follows.
	\begin{defi} \rm \label{defi-similarityindexes}
		For two similar admissible behaviors $\mathcal{B}_{1,x_1}$ and $\mathcal{B}_{2,x_2}$, let them be decomposed as (\ref{eq-decomposition}), the similarity indexes between $\mathcal{B}_{1,x_1}$ and $\mathcal{B}_{2,x_2}$, denoted by $\textbf{SI}\left(\mathcal{B}_{1,x_1},\mathcal{B}_{2,x_2}\right)$, refer to the cosine of the principal angles between $\mathcal{W}_1$ and $\mathcal{W}_2$, i.e.,
		\begin{equation}\nonumber
			\textbf{SI}\left(\mathcal{B}_{1,x_1},\mathcal{B}_{2,x_2}\right):=\cos\Theta\left(\mathcal{W}_1,\mathcal{W}_2\right).
		\end{equation}
	\end{defi}
	\begin{rem} \rm
		The definition of the similarity indexes between two admissible behaviors can be alternatively interpreted via a geometric perspective. Since similar admissible behaviors are essentially intersecting affine hyperplanes in Euclidean space, two affine hyperplanes in the $3$-dimentional Euclidean space $\mathbb{R}^3$ may serve as the illustrative examples, as depicted in Figs. \ref{fig-lesssimilar} and \ref{fig-moresimilar}. The principal angles between the intersecting affine hyperplanes in Figs. \ref{fig-lesssimilar} and \ref{fig-moresimilar} are denoted as $\Phi'$ and $\Phi''$, respectively. It is readily observed that $\cos\left(\Phi''\right)$ is more closer to $\begin{bmatrix}
			1,\ 1
		\end{bmatrix}$ compared to $\cos\left(\Phi'\right)$. By leveraging Definition \ref{defi-similarityindexes}, the affine hyperplanes in Fig. \ref{fig-moresimilar} are more similar.
		\begin{figure}[!ht]
			\centering
			\includegraphics[width=0.6\columnwidth]{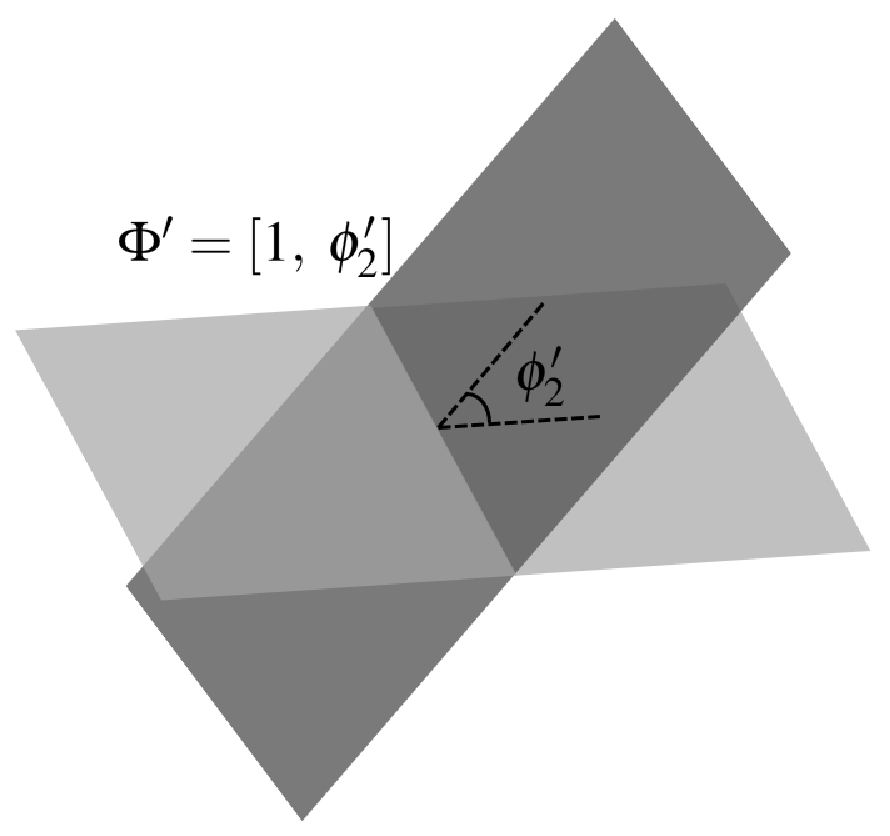}
			\caption{Less similar affine hyperplanes in $\mathbb{R}^3$.}
			\label{fig-lesssimilar}
		\end{figure}
		\begin{figure}[!ht]
			\centering
			\includegraphics[width=0.65\columnwidth]{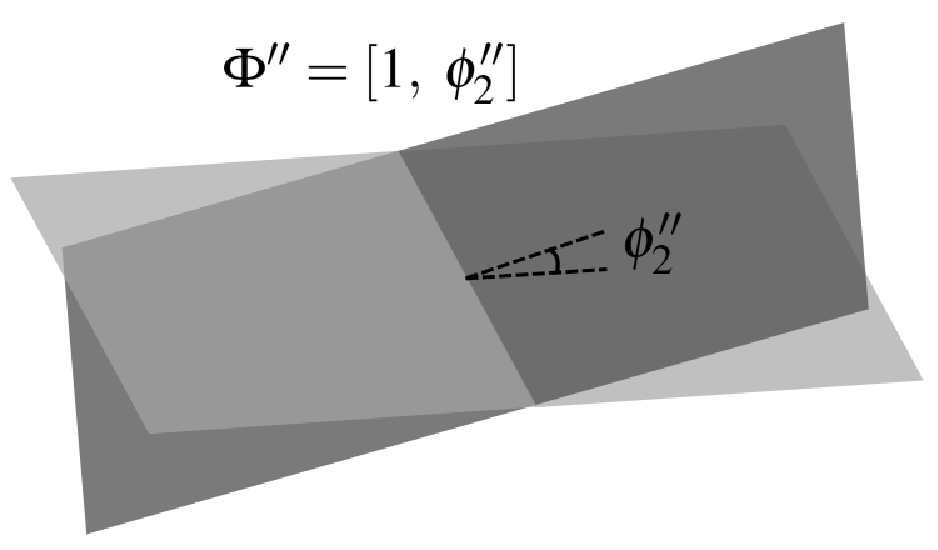}
			\caption{More similar affine hyperplanes in $\mathbb{R}^3$.}
			\label{fig-moresimilar}
		\end{figure}
	\end{rem}
	\begin{rem} \rm
		If some elements of the similarity indexes equal to zero, i.e., $s_k=s_{k+1}=\cdots=s_{n_uT}=0$, there exist two sequences $\{x_k,x_{k+1},\cdots,x_{n_uT}\},\ x_i\in\mathcal{W}_1$ and $\{y_k,y_{k+1},\cdots,y_{n_uT}\},\ y_i\in\mathcal{W}_2$ satisfying the following properties:
		\begin{enumerate}
			\item[P1)] $\langle x_i,x_j\rangle=0,\ \langle y_i,y_j \rangle=0$ for all $i,j\in\mathbb{Z}_{n_uT}\backslash\mathbb{Z}_{k-1}$;
			\item[P2)] $\langle x_i,y_i \rangle=0$ for all $i\in\mathbb{Z}_{n_uT}\backslash\mathbb{Z}_{k-1}$.
			\end{enumerate}
		Moreover, there exist two sequences $\{w_{1,k}-w_{1,\text{off}},w_{1,k+1}-w_{1,\text{off}},\cdots,w_{1,n_uT}-w_{1,\text{off}}\}$ and $\{w_{2,k}-w_{2,\text{off}},w_{2,k+1}-w_{2,\text{off}},\cdots,w_{2,n_uT}-w_{2,\text{off}}\}$ where $w_{1,i}\in\mathcal{B}_{1,x_1}$ and $w_{2,i}\in\mathcal{B}_{2,x_2}$. These two sequences satisfy the properties P1) and P2).
	\end{rem}
	Inheriting the properties of principal angles, similarity indexes are likewise independent on offset components. Nevertheless, Definition \ref{defi-similarityindexes} is not practical for computing similarity indexes between admissible behaviors. Let the SVD of the matrix $H_1^{\rm T}H_2$ be defined as
	\begin{equation} \label{eq-SVD}
		H_1^{\rm T}H_2=UDV^{\rm T}
	\end{equation}
	where
	\begin{equation} \label{eq-singularvalues}
		D=\text{diag}\left(s_1,s_2,\cdots,s_{n_uT}\right),\ s_1\geq s_2\geq \cdots\geq s_{n_uT}\geq0.
	\end{equation}
	The following theorem provides an efficient strategy for calculating the similarity indexes via the SVD.
	\begin{thm} \rm \label{thm-principalangle}
		For similar admissible behaviors $\mathcal{B}_{1,x_1}$ and $\mathcal{B}_{2,x_2}$, let them be decomposed as (\ref{eq-decomposition}). Let the SVD of $H_1^{\rm T}H_2$ be given as (\ref{eq-SVD}) and (\ref{eq-singularvalues}).
		Then the similarity indexes are calculated as
		\begin{equation}
			\textbf{SI}\left(\mathcal{B}_{1,x_1},\mathcal{B}_{2,x_2}\right)=\begin{bmatrix}
				s_1,s_2,\cdots,s_{n_uT}
			\end{bmatrix}.
		\end{equation}
		Moreover, the principal vectors associated with $\mathcal{W}_{1}$ and $\mathcal{W}_{2}$ are given by the columns of $H_1U$ and $H_2V$.
	\end{thm}
	\begin{proof}
		From the definitions of singular values and singular vectors of the matrix $H_1^{\rm T}H_2$, the $k$-th largest singular value of $H_1^{\rm T}H_2$, denoted by $s_k$, can be calculated by
		\begin{equation}\nonumber
			s_k=\max_{\left\| l\right\|=\left\| v\right\|=1}l^{\rm T}H_1^{\rm T}H_2v=l_k^{\rm T}H_1^{\rm T}H_2v_k,\ k\in\mathbb{Z}_{n_uT}\backslash\{0\}		\end{equation}
			subject to
			\begin{equation} \nonumber
				\langle l,l_i\rangle=\langle v,v_i\rangle=0,\ i\in\mathbb{Z}_{k-1}\backslash\{0\}
			\end{equation}
			where $l_i\in\mathbb{R}^{n_uT}$ and $v_i\in\mathbb{R}^{n_uT}$. Since the matrices $H_1$ and $H_2$ constructed in (\ref{eq-decomposition}) are both orthogonal matrices, by introducing the following coordinate transformation
			\begin{equation} \label{eq-coodinate transformation}
				\begin{aligned}
					x_i&=H_1l_i\in\mathcal{W}_1,\ x=H_1l\in\mathcal{W}_1\\
					y_i&=H_2v_i\in\mathcal{W}_2,\ y=H_2v\in\mathcal{W}_2
				\end{aligned},\ i\in\mathbb{Z}_{k-1}\backslash\{0\}
			\end{equation}
			it follows that
			\begin{equation} \nonumber
				\begin{aligned}
					\left\| x\right\|&=\left\| H_1l\right\|=\left\| l\right\|=1,\\
					\left\| y\right\|&=\left\| H_2v\right\|=\left\| v\right\|=1
				\end{aligned}
			\end{equation}
			and
			\begin{equation}\nonumber
				\begin{aligned}
					\langle x,x_i\rangle&=\langle H_1l,H_1l_i\rangle=\langle l,l_i\rangle=0,\\
					\langle y,y_i\rangle&=\langle H_2v,H_2v_i\rangle=\langle v,v_i\rangle=0
				\end{aligned}
				,\ i\in\mathbb{Z}_{k-1}\backslash\{0\}.
			\end{equation}
			Therefore, the singular values of $H_1^{\rm T}H_2$ can be alternatively represented by
			\begin{equation}\nonumber
				s_k=\max_{x\in\mathcal{W}_1}\max_{y\in\mathcal{W}_2}\langle x,y\rangle=\langle x_k,y_k\rangle,\ k\in\mathbb{Z}_{n_uT}\backslash\{0\}
			\end{equation}
			subject to
			\begin{equation} \nonumber
				\left\| x\right\|=\left\| y\right\|=1,\ \langle x,x_i\rangle=0,\ \langle y,y_i\rangle=0,\ i\in\mathbb{Z}_{k-1}\backslash\{0\}.
			\end{equation}
			From Definition \ref{defi-principalangle}, it is obvious that the $k$-th largest singular value of the matrix $H_1^{\rm T}H_2$ is exactly the cosine of $k$-th smallest principal angle between the subspaces $\mathcal{W}_1$ and $\mathcal{W}_2$ spanned by the columns of $H_1$ and $H_2$, that is,
			\begin{equation} \nonumber
				s_k=\cos\left(\theta_k\right),\ k\in\mathbb{Z}_{n_uT}\backslash\{0\}.
			\end{equation}
			Therefore, the similarity indexes can be efficiently obtained by calculating the singular values of the matrix $H_1^{\rm T}H_2$. Moreover, by leveraging Definition \ref{defi-principalangle} and noticing that the vectors $l_i$ and $v_i$ are exactly the $i$-th column of the orthogonal matrices $U$ and $V$, that is,
			\begin{equation} \nonumber
				\begin{aligned}
					U&=\begin{bmatrix}
						l_1,\ l_2,\ \cdots,\ l_{n_uT}
					\end{bmatrix},\\ V&=\begin{bmatrix}
						v_1,\ v_2,\ \cdots,\ v_{n_uT}
					\end{bmatrix}
				\end{aligned}
			\end{equation}
			the principal vectors associated with the subspaces $\mathcal{W}_{1}$ and $\mathcal{W}_{2}$ can be obtained from the coordinate transformation (\ref{eq-coodinate transformation}) as
			\begin{equation} \nonumber
				\begin{aligned}
					\begin{bmatrix}
						x_1,\ x_2,\ \cdots,\ x_{n_uT}
					\end{bmatrix}&=H_1\begin{bmatrix}
						l_1,\ l_2,\ \cdots,\ l_{n_uT}
					\end{bmatrix}=H_1U, \\
					\begin{bmatrix}
						y_1,\ y_2,\ \cdots,\ y_{n_uT}
					\end{bmatrix}&=H_2\begin{bmatrix}
						v_1,\ v_2,\ \cdots,\ v_{n_uT}
					\end{bmatrix}=H_2V.
				\end{aligned}
			\end{equation}
			So far, we have proved that the similarity indexes and principal vectors can be efficiently calculated via the SVD of $H_1^{\rm T}H_2$.
	\end{proof}
	\begin{rem} \rm
		By leveraging the similarity indexes, we can quantitatively evaluate the similarity between admissible behaviors of heterogeneous systems. Particularly, for similar admissible behaviors, if the similarity indexes are $\cos\Theta\left(\mathcal{W}_1,\mathcal{W}_2\right)=1^{\rm T}_{n_uT}$, it indicates that two systems $\Sigma_{1,\mathbb{T}}$ and $\Sigma_{2,\mathbb{T}}$ have identical admissible behaviors. In this case, the successful experience for any control tasks of the guest system can be directly implemented to the host system to address Problem \ref{pro-2}), while ensuring that $\left\| w_g-w_h\right\|=0$. Otherwise, a similarity-based learning control strategy needs to be proposed to minimize the difference between $w_g$ and $w_h$.
	\end{rem}

	\section{Similarity-Based Learning Via Projection}
	\label{sec:similarity-based learning via projection}
	In this section, we focuses on dealing with the similarity-based learning control problems, i.e., Problem \ref{pro-2}), by leveraging the similarity indexes between two admissible behaviors of heterogeneous systems. We suppose that, through some powerful control strategies, the guest system $\Sigma_{2,\mathbb{T}}$ has already achieved the desired admissible trajectory $w_g\in\mathcal{B}_{2,x_2}$, ensuring the accomplishment of its control task. When accomplishing the same task for the host system $\Sigma_{1,\mathbb{T}}$, the successful experience of $\Sigma_{2,\mathbb{T}}$ can provide helpful guidance, and the reliability of this guidance depends on the similarity indexes between the admissible behaviors of two systems. The admissible behaviors are essential high-dimensional affine hyperplanes, and we can once again utilize the affine hyperplanes in the $3$-dimensional Euclidean space as a simple example, as depicted in Fig. \ref{fig-similarlearning}. Although this example slightly loses rigor, it can be employed to illustrate the machenism of the similarity-based learning control. That is, after obtaining the desired admissible trajectory $w_g\in\mathcal{B}_{2,x_2}$, the similarity-based learning control aims at directly deriving $w_h\in\mathcal{B}_{2,x_2}$ minimizing $\left\|w_g-w_h\right\|$ by leveraging $w_g$ and the similarity indexes $\cos\Phi$. Obviously, the projection techniques may serve as helpful tools.
	\begin{figure}[!ht]
		\centering
		{\includegraphics[width=\columnwidth]{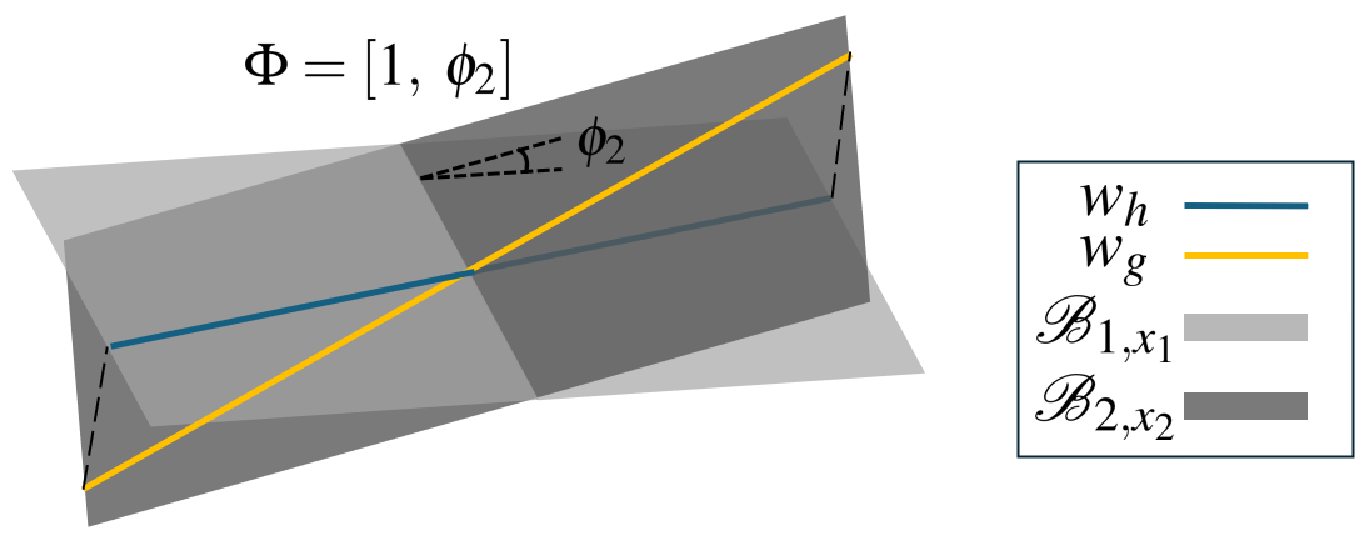}}
		\caption{An example for illustrating the machenism of similarity-based learning control}
		\label{fig-similarlearning}
	\end{figure}
	In the following context, we provide a detailed demonstration on how to integrate the successful experience of $\Sigma_{2,\mathbb{T}}$ into the learning control of $\Sigma_{1,\mathbb{T}}$.
	
	We define the orthogonal projection operator onto $\mathcal{W}_1$ as $P_{\mathcal{W}_1}\left(\cdot\right)$ which is a linear operator. Additionally, the orthogonal projection operator onto $\mathcal{B}_{1,x_1}$ is defined as $P_{\mathcal{B}_{1,x_1}}(\cdot)$, that is, for some $x\in\mathbb{R}^{n_wT}$, $P_{\mathcal{B}_{1,x_1}}(x)\in\mathcal{B}_{1,x_1}$ represents the orthogonal projection of $x$ onto $\mathcal{B}_{1,x_1}$, ensuring $\left\|x-P_{\mathcal{B}_{1,x_1}}(x)\right\|$ is minimized. However, of note is that this projection operator is not a linear one, since $\mathcal{B}_{1,x_1}$ is not a subspace. To address this difficulty, we introduce the following lemma.
	\begin{lem} \rm \cite{JPlesnik2007LAA} \label{lem-affine2subspace}
		Let the admissible behavior $\mathcal{B}_{1,x_1}$ be decomposed as (\ref{eq-decomposition}), where $\mathcal{W}_1$ is a subspace, then
		\begin{equation} \nonumber
			P_{\mathcal{B}_{1,x_1}}(x)=w_{1,\text{off}}+P_{\mathcal{W}_1}(x-w_{1,\text{off}})
		\end{equation}
		holds for all $x\in\mathbb{R}^{n_wT}$.
	\end{lem}
	Indeed, what Problem 2) looks for, i.e., $w_h$, is the orthogonal projection of $w_g$ onto $\mathcal{B}_{1,x_1}$. Lemma \ref{lem-affine2subspace} calculates this orthogonal projection by investigating another orthogonal projection onto an associated subspace. If $P_{\mathcal{W}_1}$$\left(x-w_{1,\text{off}}\right)$ is calculated by utilizing the model information of $\mathcal{W}_1$, then we fall into the existing framework of learning-based control. The superiority of the similarity-based learning control strategy proposed in this paper lies in the fact that this projection can be efficiently obtained by employing the similarity indexes and principal vectors, without relying on the specific information of $\mathcal{W}_1$ and existing learning-based control strategies. The following theorem elaborates on this claim in detail.
	\begin{thm} \rm \label{thm-similarity-basedlearning}
		For similar admissible behaviors $\mathcal{B}_{1,x_1}$ and $\mathcal{B}_{2,x_2}$, let them be decomposed as (\ref{eq-decomposition}), and let the SVD of $H_1^{\rm T}H_2$ be given as (\ref{eq-SVD}) and (\ref{eq-singularvalues}).
		If the desired admissible trajectory is $w_g\in\mathcal{B}_{2,x_2}$, then the optimal admissible trajectory $w_h\in\mathcal{B}_{1,x_1}$ is calculated as
		\begin{equation} \label{eq-similarity-basedlearning}
			w_h=H_1UD\overline{g}+P_{\mathcal{W}_1}\left(w_{2,\text{off}}-w_{1,\text{off}}\right)+w_{1,\text{off}}
		\end{equation}
		where $\overline{g}$ satisfies
		\begin{equation}  \label{eq-calculateg}
			w_g=H_2V\overline{g}+w_{2,\text{off}}.
		\end{equation}
		In this situation, the difference $\left\| w_h-w_g\right\|$ is minimized, or equivalently, Problem \ref{pro-2}) is addressed.
	\end{thm}
	\begin{proof}
		Thanks to the fact that the admissible behaviors $\mathcal{B}_{1,x_1}$ and $\mathcal{B}_{2,x_2}$ are similar, by leveraging the decomposition in (\ref{eq-decomposition}) and the results presented in Theorem \ref{thm-principalangle}, the similarity indexes between $\mathcal{B}_{1,x_1}$ and $\mathcal{B}_{2,x_2}$ can be calculated through the SVD of $H_1^{\rm T}H_2$ as
		\begin{equation} \nonumber
			\textbf{SI}\left(\mathcal{B}_{1,x_1},\mathcal{B}_{2,x_2}\right)=\begin{bmatrix}
				s_1,s_2,\cdots,s_{n_uT}
			\end{bmatrix}.
		\end{equation}
		Moreover, the columns of $H_1U$ and $H_2V$ serve as the principal vectors associated with the subspaces $\mathcal{W}_1$ and $\mathcal{W}_2$. As we have emphasized earlier, the to-be-sought optimal admissible trajectory $w_h\in\mathcal{B}_{1,x_1}$ is exactly the orthogonal projection of $w_g$ onto $\mathcal{B}_{1,x_1}$, i.e., $P_{\mathcal{B}_{1,x_1}}\left(w_g\right)$. From the decomposition of $\mathcal{B}_{i,x_i}$ constructed in (\ref{eq-decomposition}), by noticing the fact $w_g\in\mathcal{B}_{2,x_2}$, there must exists some vector $\overline{g}$ satisfying
		\begin{equation}\nonumber
			w_g=H_2V\overline{g}+w_{2,\text{off}}.
		\end{equation}
		From (\ref{eq-calculateg}), the to-be-sought $w_h$ can be expressed as
		\begin{equation} \label{eq-whcalculation1}
			\begin{aligned}
				w_h&=P_{\mathcal{B}_{1,x_1}}\left(w_g\right) \\
				&=P_{\mathcal{B}_{1,x_1}}\left(H_2V\overline{g}+w_{2,\text{off}}\right).
			\end{aligned}
		\end{equation}
		Since the operator $P_{\mathcal{B}_{1,x_1}}\left(\cdot\right)$ is not a linear operator, directly calculating the orthogonal projection in (\ref{eq-whcalculation1}) is challenging. By leveraging Lemma \ref{lem-affine2subspace}, this orthogonal projection can be derived through calculating corresponding orthogonal projection onto the subspace component as
		\begin{equation} \label{eq-proof2thm2}
			\begin{aligned}
				w_h
				&=w_{1,\text{off}}+P_{\mathcal{W}_1}\left(H_2V\overline{g}+w_{2,\text{off}}-w_{1,\text{off}}\right).
			\end{aligned}
		\end{equation}
		Since the operator $P_{\mathcal{W}_1}\left(\cdot\right)$ is a linear one, from (\ref{eq-proof2thm2}), the to-be-sought $w_h$ can be further expressed as
		\begin{equation} \nonumber
			w_h=P_{\mathcal{W}_1}\left(H_2V\right)\overline{g}+P_{\mathcal{W}_1}\left(w_{2,\text{off}}-w_{1,\text{off}}\right)+w_{1,\text{off}}.
		\end{equation}
		Afterward, we elaborate that $P_{\mathcal{W}_1}(H_2V)$ can be efficiently derived by implementing the similarity indexes between $\mathcal{B}_{1,x_1}$ and $\mathcal{B}_{2,x_2}$. Let the $i$-th column of $H_1U$ (or $H_2V$) be denoted as $\left(H_1U\right)_i$ (or $\left(H_2V\right)_i$), then $P_{\mathcal{W}_1}\left(H_2V\right)$ can be expressed as
		\begin{equation} \nonumber
			P_{\mathcal{W}_1}\left(H_2V\right)=
			\begin{bmatrix}
				P_{\mathcal{W}_1}\left(H_2V\right)_1,\ P_{\mathcal{W}_1}\left(H_2V\right)_2,\ \cdots,\ P_{\mathcal{W}_1}\left(H_2V\right)_{n_uT}
			\end{bmatrix}
		\end{equation}
		where
		\begin{equation} \nonumber
			P_{\mathcal{W}_1}\left(H_2V\right)_i=\sum_{j=1}^{n_uT}\langle (H_2V)_i,(H_1U)_j\rangle(H_1U)_j,\ \forall i\in\mathbb{Z}_{n_uT}\backslash\{0\}.
		\end{equation}
		From the definitions of similarity indexes and principal vectors and the properties of SVD, the following two properties hold:
		\begin{equation} \nonumber
			\begin{aligned}
				s_k&=\langle (H_2V)_k,(H_1U)_k\rangle,\ \forall k\in\mathbb{Z}_{n_uT}\backslash\{0\},\\
				0&=\langle (H_2V)_i,(H_1U)_j\rangle,\ \forall i\neq j,\ \forall i,j\in\mathbb{Z}_{n_uT}\backslash\{0\}.
			\end{aligned}
		\end{equation}
		Building upon the above properties, the orthogonal projection $P_{\mathcal{W}_1}\left(H_2V\right)$ can be rewritten as
		\begin{equation}\nonumber
			\begin{aligned}
				&P_{\mathcal{W}_1}(H_2V)\\
				=&\begin{bmatrix}
					P_{\mathcal{W}_1}\left(H_2V\right)_1,\ P_{\mathcal{W}_1}\left(H_2V\right)_2,\ \cdots,\ P_{\mathcal{W}_1}\left(H_2V\right)_{n_uT}
				\end{bmatrix}\\
				=&\begin{bmatrix}
					(H_1U)_1s_1,\ (H_1U)_2s_2,\ \cdots,\ (H_1U)_{n_uT}s_{n_uT}
				\end{bmatrix}\\
				=&H_1UD
			\end{aligned}
		\end{equation}
		and the optimal admissible trajectory $w_h\in\mathcal{B}_{1,x_1}$ can be further represented as
		\begin{equation} \nonumber
			w_h=H_1UD\overline{g}+P_{\mathcal{W}_1}\left(w_{2,\text{off}}-w_{1,\text{off}}\right)+w_{1,\text{off}}.
		\end{equation}
		Since $w_h$ is exactly the orthogonal projection of $w_g$ onto $\mathcal{B}_{1,x_1}$, the difference $\left\| w_h-w_g\right\|$ is naturally minimized, and the proof is completed.
	\end{proof}
	\begin{rem} \rm
		The significance of Theorem \ref{thm-similarity-basedlearning} lies in accomplishing the learning control tasks for the host system by utilizing the similarity indexes and successful experience from the guest system. In (\ref{eq-similarity-basedlearning}), the successful experience can be characterized by the vector $\overline{g}$, and what we need to do is to calculate $\overline{g}$ from $w_g$ via (\ref{eq-calculateg}). This strategy possesses excellent advantages when dealing with multiple tasks, i.e., different desired admissible trajectory $w_g$, since the terms $P_{\mathcal{W}_1}\left(w_{2,\text{off}}-w_{1,\text{off}}\right)$ and $w_{1,\text{off}}$ are task-independent. Specifically, all we need to do is to calculate the corresponding $\overline{g}$ for different $w_g$. This allows us to directly calculate the optimal admissible trajectory $w_h$ for the host system without identification of model knowledge and repeatedly resorting to existing learning control methods.
	\end{rem}
	\begin{rem} \rm
		Additionally, as emphasized in Remark \ref{rem-datadriven}, the similarity-based learning control strategy proposed in this paper can immediately reconstructed through data-driven methods. Indeed, the advantages of the similarity-based learning control strategy are even more pronounced within a data-driven framework. The specific system model knowledge $(A_i{t},B_i(t),C_i(t),D_i(t))$ is no longer required to be identified from data. Instead, only the subspace and offset components, $\mathcal{W}_i$ and $w_{i,\text{off}}$, and the similarity indexes need to be identified, and the learning control for the host system can be directly completed. Compared to the existing data-driven learning control strategies, similarity-based learning control significantly improves the efficiency of learning control.
	\end{rem}

	\section{Simulation Examples} \label{sec:Simulation Examples}
	In this section, numerical examples are provided to illustrate the effectiveness of the similarity-based learning control strategy proposed in this paper. We present two numerical examples. Example \ref{exam-multipletasks} focuses on demonstrating the effectiveness of similarity-based learning control strategy when dealing with multiple tasks, while Example \ref{exam-differentsimilarityindexes} aims at illustrating the impact of similarity indexes on the effectiveness of similarity-based learning control.
	\subsection{Example 1: Effectiveness in Multiple Tasks}
	\label{exam-multipletasks}
	In this subsection, we provide an example to demonstrate the effectiveness of the similarity-based learning control in multiple tasks. The model information regarding the host and guest systems, represented by $\Sigma_{1,\mathbb{T}}$ and $\Sigma_{2,\mathbb{T}}$, is given in (\ref{eq-simulation1}).
	\begin{equation}\label{eq-simulation1}
		\begin{aligned}
			A_1(t)&=\begin{bmatrix}
				0.05t  & 1      & 0\\
				0      & 0.05t  & 1\\
				-0.500 & -1.850 & -2.500+0.05t
			\end{bmatrix}, \\
			A_2(t)&=\begin{bmatrix}
				0.05t  & 1      & 0\\
				0      & 0.05t  & 1\\
				-0.512 & -1.920 & -2.400+0.05t
			\end{bmatrix},\\
			B_1(t)&=B_2(t)=\begin{bmatrix}
				6\\
				0\\
				0.500
			\end{bmatrix},\ C_1(t)=C_2(t)=\begin{bmatrix}
				2\\
				\sqrt{2}\\
				0
			\end{bmatrix}^{\rm T},\\
			D_1(t)&=D_2(t)=0,\ x_1(0)=\begin{bmatrix}
				0\\
				0\\
				1.020
			\end{bmatrix},\ x_2(0)=\begin{bmatrix}
				0\\
				0\\
				1
			\end{bmatrix}.
		\end{aligned}
	\end{equation}
	This situation typically arises when there is model uncertainty between the host system and the guest system. We assume that both the host and guest systems are confronted with common output tracking tasks, where they need to track two different reference output, as shown in (\ref{eq-reference}), over a given time duration $\mathbb{T}=\mathbb{Z}_{24}$.
	\begin{equation}\label{eq-reference}
		\begin{aligned}
			r^1(t)&=\sin\left(\frac{\pi}{4}t\right),\ \forall t\in\mathbb{Z}_{24},\\
			r^2(t)&=\left\{\begin{aligned}
				1,\ t \text{ mod } 8\in\{1,2,3,4\}\\
				0,\ t \text{ mod } 8\in\{0,5,6,7\}
			\end{aligned}\right.,\ \forall t\in\mathbb{Z}_{24}.
		\end{aligned}
	\end{equation}
	To address the output tracking issues of the guest system, ILC may serve as a powerful tool. By leveraging the ILC algorithm proposed in \cite{DHOwens2005ARC}, output tracking tasks of the guest system are accomplished after 500 algorithm iterations. When dealing with the tracking issues of the host system, there is no need to reapply the ILC algorithm. Instead, we can directly utilize the learning experience of the guest system to facilitate the learning control of the host system via the similarity-based learning control strategy proposed in this paper. For the desired reference $r^1(t)$, the learned inputs and outputs of the guest system and host systems are depicted in Fig. \ref{fig-ref1}. In Fig. \ref{fig-ref1}, $u_g^1(t)$-ILC and $y_g^1(t)$-ILC represent the learned input and output for the guest system via ILC. By leveraging the successful experience of the guest system, $u_h^1(t)$-SBL and $y_h^1(t)$-SBL represent the learned input and output for the host system via similarity-based learning control. Likewise, with the application of the similarity-based learning control strategy, the learned inputs and outputs of guest system and host system for the reference $r^2(t)$ are shown in Fig. \ref{fig-ref2}.
	\begin{figure}[!ht]
		\centering
		{\includegraphics[width=\columnwidth]{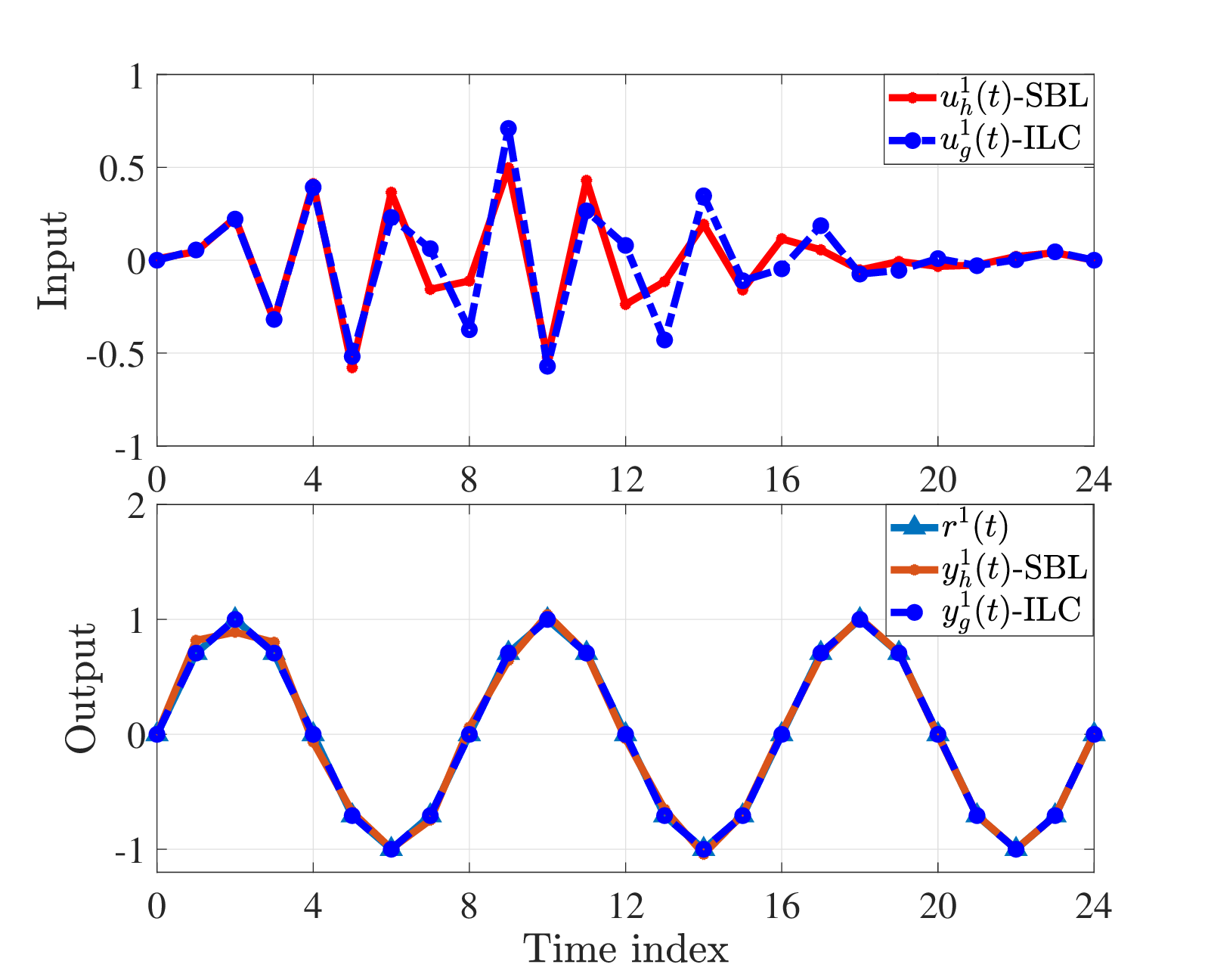}}
		\caption{Inputs and outputs of $\Sigma_{1,\mathbb{T}}$ and $\Sigma_{2,\mathbb{T}}$ for reference $r^1(t)$.}
		\label{fig-ref1}
	\end{figure}
	\begin{figure}[!ht]
		\centering
		{\includegraphics[width=\columnwidth]{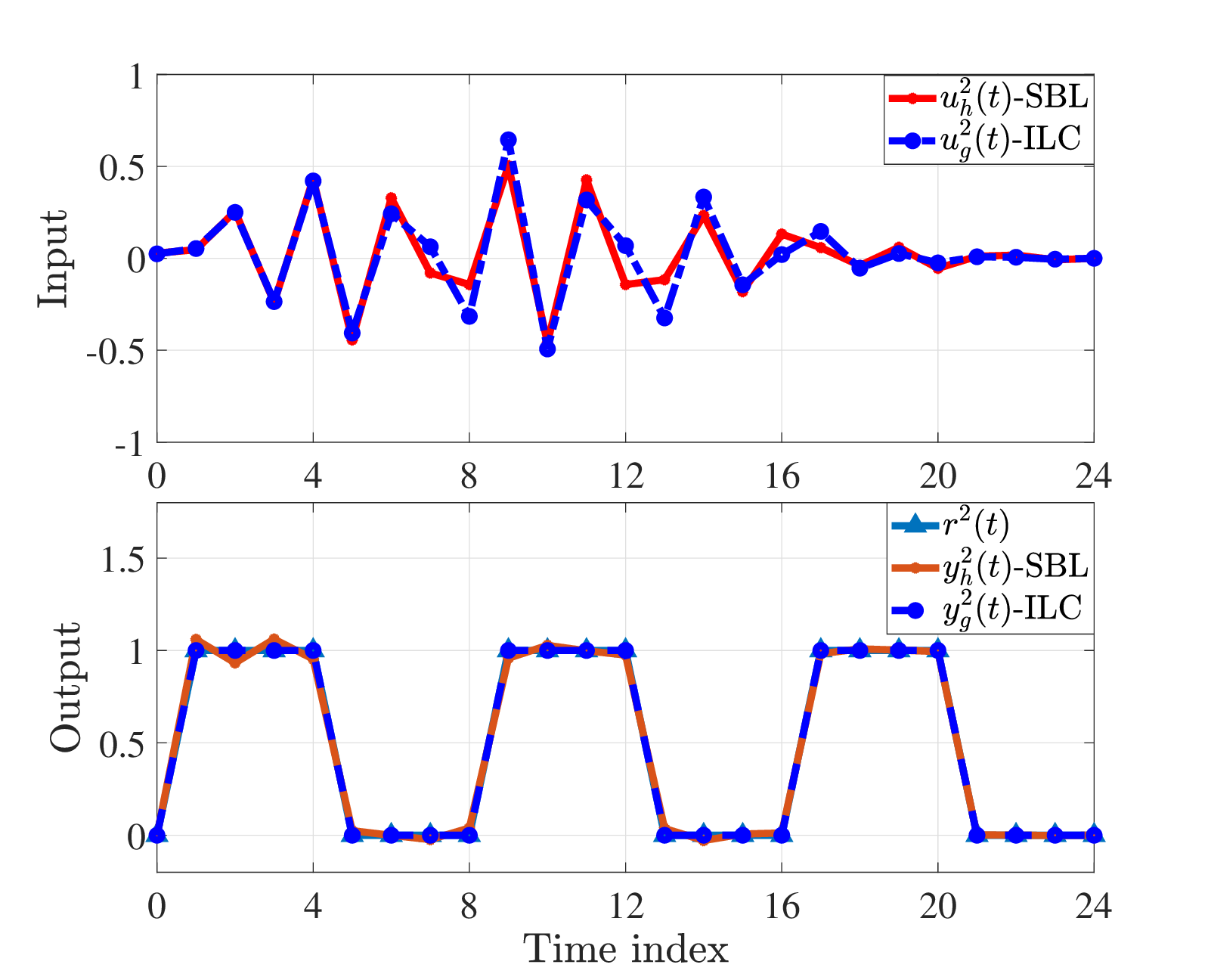}}
		\caption{Inputs and outputs of $\Sigma_{1,\mathbb{T}}$ and $\Sigma_{2,\mathbb{T}}$ for reference $r^2(t)$.}
		\label{fig-ref2}
	\end{figure}
	
	Through this numerical example, it is observed that with the guest system leveraging ILC to accomplish the output tracking tasks, the host system can employ the similarity-based learning control without the need for any other learning control algorithm to directly achieve the same output tracking tasks. Therefore, the effectiveness of similarity-based learning control in multiple tasks has been validated.
	\subsection{Example 2: Effectiveness with Different Similarity Indexes}
	\label{exam-differentsimilarityindexes}
	Building upon Example \ref{exam-multipletasks}, this numerical example is employed to demonstrate the fact that the successful experience of guest systems with high similarity, which is characterized by similarity indexes, is more benificial. We consider the output tracking task for the desired reference $r^1(t)$ in Example \ref{exam-multipletasks}, with the difference being the introduction of an additional guest system $\Sigma_{3,\mathbb{T}}$, whose model information is given as
	\begin{equation}
		\begin{aligned}
			A_3(t)&=\begin{bmatrix}
				0.05t  & 1      & 0\\
				0      & 0.05t  & 1\\
				-0.600 & -2     & -2.300+0.05t
			\end{bmatrix}\\
			B_3(t)&=\begin{bmatrix}
				6 \\
				0 \\
				0.500
			\end{bmatrix},
			C_3(t)=\begin{bmatrix}
				2\\
				\sqrt{2}\\
				0
			\end{bmatrix}^{\rm T},D_3(t)=0.
		\end{aligned}
	\end{equation}
	and the initial state is $x_3(0)=\begin{bmatrix}
		0,\ 0,\ 1.1
	\end{bmatrix}^{\rm T}$.
	With the application of Theorem \ref{thm-principalangle}, it is concluded that $\mathcal{B}_{2,x_2}$ is more similar to $\mathcal{B}_{1,x_1}$ than $\mathcal{B}_{3,x_3}$. By leveraging ILC algorithm, the guest system $\Sigma_{3,\mathbb{T}}$ can achieve perfect tracking performance for the reference $r_1(t)$. With the successful experience of $\Sigma_{3,\mathbb{T}}$, the tracking task of $\Sigma_{1,\mathbb{T}}$ can be accomplished with the similarity-based learning control. In Fig. \ref{fig-differentsystem}, $u_g^3(t)$-ILC and $y_g^3(t)$-ILC represent the learned input and output for the guest system $\Sigma_{3,\mathbb{T}}$ via ILC, and $u^3_h(t)$-SBL and $y^3_h(t)$-SBL denote the learned input and output for the host system $\Sigma_{1,\mathbb{T}}$.
		\begin{figure}[!ht]
		\centering
		{\includegraphics[width=\columnwidth]{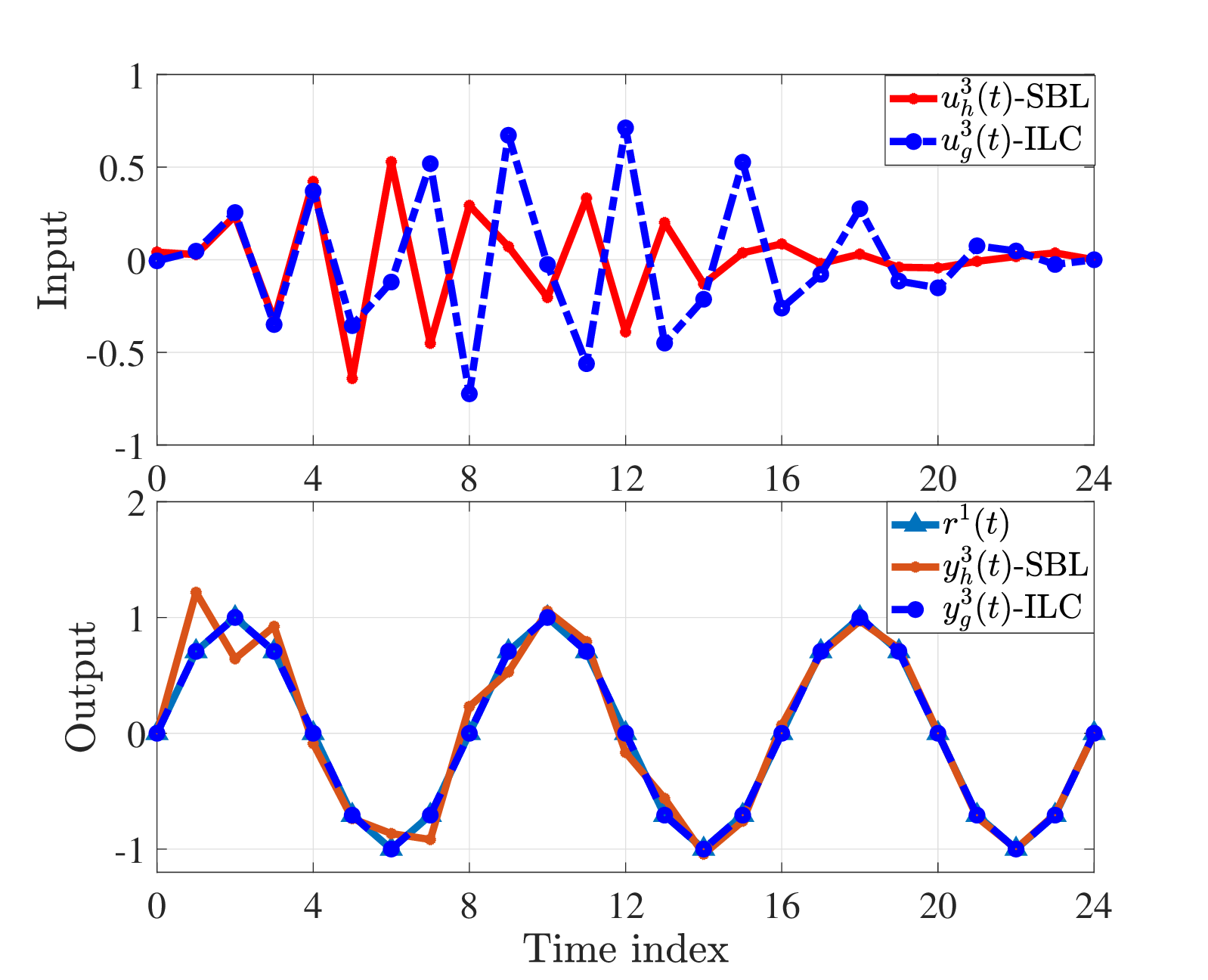}}
		\caption{Inputs and outputs of $\Sigma_{1,\mathbb{T}}$ and $\Sigma_{3,\mathbb{T}}$ for reference $r^1(t)$.}
		\label{fig-differentsystem}
	\end{figure}
	By comparing the control performances shown in Figs. \ref{fig-ref1} and \ref{fig-differentsystem}, it is concluded that the successful experience of a more similar guest system is more beneficial for the learning control of the host system.
	\section{Conclusions and Future Works} \label{sec:Conclusions and future works}
	In this paper, the definitions of similarity and similarity-indexes for heterogeneous linear systems have been proposed, based on which an innovative similarity-based learning control strategy has been developed. By defining the similarity indexes between admissible behaviors as the principal angles between their subspace components, we have quantitatively assessed the value of the successful experience of the guest systems for the host system. Furthermore, by exploiting the projection techniques, a similarity-based learning control strategy has been proposed. Based on this, the host system has accomplished the control tasks by directly leveraging the successful experience of the guest system, without repeatedly resorting to any sophisticate learning-based control frameworks.

	There are still many issues worth exploring in the domain of similarity-based learning control. One of these issues is to reformulate the proposed similarity-based learning control strategy in a data-driven framework, in order to better highlight its advantages over other learning-based control frameworks. Additionally, the proposed similarity-based learning control is an primary and open-loop control strategy. Specifically, the difference between $w_h$ and $w_g$ depends on the similarity indexes, thus rendering it unadjustable. Therefore, in the future, it is necessary to further develop a closed-loop similarity-based learning control strategy to allow for the precise control of the difference $\left\| w_g-w_h\right\|$.
	\section*{Acknowledgment}
	The authors would like to thank Prof. Kevin L. Moore from the Department of Electrical Engineering at the Colorado School of Mines for his insightful remarks and discussions on the results presented in this paper.

	\bibliographystyle{IEEEtran}
	\bibliography{IEEEabrv,myref}

\end{document}